\documentclass{article}
\usepackage{amssymb}
\usepackage{url}
\usepackage{amsmath}
\usepackage{mathrsfs}
\usepackage{enumerate}
\usepackage{amsthm}
\usepackage{xspace}
\usepackage[utf8]{inputenc}
\usepackage{algorithmicx}
\usepackage{algpseudocode}
\usepackage{tikz}
\usepackage{graphicx}
\usetikzlibrary{shapes,shapes.geometric,arrows,automata}
\usepackage{graphicx}

\newtheorem{lemma}{Lemma}
\newtheorem{theorem}{Theorem}
\newtheorem{corollary}{Corollary}
\newtheorem{example}{Example}
\newtheorem{definition}{Definition}

\newcommand{\N}{{\mathbb N}}

\newcommand{\card}[1]{|#1|}
\newcommand{\dfa}{DFA\xspace}
\newcommand{\nfa}{NFA\xspace}
\newcommand{\dfas}{DFAs\xspace}
\newcommand{\nfas}{NFAs\xspace}
\newcommand{\comp}[1]{\overline{#1}}
\newcommand{\suff}[1]{\mathsf{suff}(#1)}
\newcommand{\preff}[1]{\mathsf{pref}(#1)}
\newcommand{\inff}[1]{\mathsf{infix}(#1)}
\newcommand{\suffo}{\mathsf{suff}}
\newcommand{\preffo}{\mathsf{pref}}
\newcommand{\inffo}{\mathsf{infix}}
\newcommand{\Succ}[1]{\mathsf{Succ}(#1)}
\newcommand{\disw}[2]{\mathsf{D}_{#1}(#2)}
\newcommand{\dis}[1]{\mathsf{D}(#1)}
\newcommand{\diso}{\mathsf{D}}
\newcommand{\disn}[2]{\mathsf{D}^{#2}(#1)}
\newcommand{\distminw}[2]{\underline{\mathsf{D}}_{#1}(#2)}
\newcommand{\distmin}[1]{\underline{\mathsf{D}}(#1)}
\newcommand{\distmino}{\underline{\mathsf{D}}}
\newcommand{\distminp}[2]{\underline{\mathsf{D}}^{#2}{(#1)}}
\newcommand{\dpre}[1]{\mathsf{E}(#1)}
\newcommand{\dpreo}{\mathsf{E}}
\newcommand{\dpremin}[1]{\underline{\mathsf{E}}(#1)}
\newcommand{\dpreminw}[2]{\underline{\mathsf{E}}_{#1}(#2)}
\newcommand{\dprew}[2]{\mathsf{E}_{#1}(#2)}
\newcommand{\dpren}[2]{\mathsf{E}^{#2}(#1)}
\newcommand{\dpremino}{\underline{\mathsf{E}}}
\newcommand{\dpreminp}[2]{\underline{\mathsf{E}}^{#2}{(#1)}}
\newcommand{\dinf}[1]{\mathsf{F}(#1)}
\newcommand{\dinfo}{\mathsf{F}}
\newcommand{\dinfmin}[1]{\underline{\mathsf{F}}(#1)}
\newcommand{\dinfminw}[2]{\underline{\mathsf{F}}_{#2}(#1)}
\newcommand{\dinfn}[2]{\mathsf{F}^{#2}(#1)}
\newcommand{\dinfw}[2]{\mathsf{F}_{#1}(#2)}

\newcommand{\notequiv}{\nLeftrightarrow}
\newcommand{\bcop}[1]{\mathsf{O}(#1)}
\newcommand{\bcopo}{\mathsf{O}}
\newcommand{\bcopn}[2]{\mathsf{O}^{#2}(#1)}

\newcommand{\mneq}[1]{\equiv_{#1}}
\newcommand{\mnleq}[1]{\eqsim_{#1}}
\newcommand{\mntseq}[1]{\approxeq_{#1}}

\date{}
\newcommand{\Set}[1]{\left\{ #1 \right\}}
\newcommand{\kleene}[1]{#1^\star}
\newcommand{\tuple}[1]{\left\langle #1\right\rangle}
\newcommand{\lang}[1]{{\cal L}(#1)}
\newcommand{\AND}{\;\wedge\:}
\title{Distinguishability Operations  and Closures on Regular Languages}
 \author{Cezar C\^ampeanu \\
Department of Computer Science and Information Technology\\
The University of Prince Edward Island, Canada\\ccampeanu@upei.ca
 \and Nelma Moreira
, Rogério Reis\\
 Centro de Matem\'atica e Faculdade de Ci\^e{}ncias da\\
   Universidade do Porto, Portugal  \\
   \{nam,rvr\}@dcc.fc.up.pt}
   
\begin{document}
\maketitle

\begin{abstract}
 Given a regular language $L$, we study the language of words $\dis{L}$, that  distinguish between pairs of different
left-quotients of $L$. We characterize this distinguishability
operation, show that its iteration has always a fixed point,
and we generalize this result to operations derived from closure operators 
and Boolean operators. 
We give an upper bound for the state complexity of the distinguishability operation, 
and prove its tightness.  
We show that the set of minimal words that can be used to distinguish 
between different left-quotients of a language $L$ has at most $n-1$ elements,
 where $n$ is the state complexity of $L$, and 
 we also study the properties of its iteration. We generalize the results for the languages of words that distinguish between pairs of different right-quotients and two-sided quotients of a language $L$.
\end{abstract}

\section{Introduction}
\label{sec:intro}
Regular languages and operations over them have been extensively
studied  during the last sixty years, the applications of automata
being continuously extended in different areas. 
As a practical example, we can use automata to model various
electronic circuits. The testing of the circuits can be done by
applying several inputs to various pins of a circuit, and checking the output
produced. 
Because in many cases the circuits emulate automata, it is useful to
develop general tools for testing various properties of automata, such
as testing the relation between the response of the circuit for the
same signal, applied to different gates. 
To answer if an automaton is minimal requires to test if two states
are equivalent or not. 
The easiest way to do this is to use as input different words, and see if
for both states, we reach states with the same finality, thus, in case
of a circuit, for both input gates we will get the same value of the output bit.  
However, checking every possible word is a tedious task, and it would
be useful to limit the testing only to the words that can distinguish
between states. 

Therefore, it is worth studying the languages that distinguish between all
non-equivalent states of a given deterministic finite automaton. 
For an automaton  $\mathcal{A}$ we denote the distinguishing  
language by $\dis{\mathcal{A}}$ and the language of minimal words
distinguishing between all non-equivalent states 
by $\distmin{\mathcal{A}}$. We can also consider the distinguishability 
languages $\dis{L}$, and $\distmin{L}$ for a regular language, $L$,
which will distinguish between all non-equivalent words.

The idea of studying word or state distinguishability is not new. 
In 1958, Ginsburg studied the length of the smallest uniform experiment 
which distinguishes the terminal states of a machine 
\cite{ginsburg58:_lengt_of_small_unifor_exper},
and with Spanier in \cite{Ginsburg1960:InputSemigroup}, he studies whether or not 
an arbitrary semigroup can serve as an input for a 
machine that distinguishes between the elements of the input semigroup.
\ A comparable work was done for terminal
distinguishability by Sempere \cite{JoseSempere}, where terminal
 segments of automata are studied to characterize language 
 families that can be identified in the limit from positive data.  
Indeed, knowing that an automaton $\mathcal{A}$ has at most $n$
states, and having the language $\distmin{\mathcal{A}}$, together with the
words of length at most $n+1$ that are in the language $\lang{\mathcal{A}}$, we
can recover the initial automaton $\mathcal{A}$.  
Note that without the language $\distmin{\mathcal{A}}$, any learning
procedure will only approximate the language $\lang{\mathcal{A}}$.  
For example, in case we know $M$ to be the set of all the
words of a language $L$ with length at most $n+1$, we can infer that
$L$ is a cover language for $M$, but we cannot determine which one of
these cover languages is $L$. 
Thus, any learning procedure would only be able to guess $L$ from $M$,
and the guess would not be accurate, as the number of cover automata
for a finite language can be staggeringly
high~\cite{campeanu13:_cover_languag_and_implem,CezarAndreiCount}. 
\ In ~\cite{
restivo12:_graph_theor_approac_to_autom_minim}, 
Restivo and
Vaglica proposed a graph-theoretical approach to test automata
minimality. For  a given automaton ${\cal A}$ they associate 
a digraph, called \emph{pair graph}, where vertices are pairs
$\Set{p,q}$ of states of ${\cal A}$, and  edges connect vertices for which
the states have a transition from the same symbol in $\mathcal{A}$. Then, two states
$p$ and $q$ of ${\cal A}$ are distinguishable if and only if there is a path
from the vertex $\Set{p,q}$ to a vertex $\Set{p', q'}$, where $p'$ is
final and $q'$  is non final, i.e., there exists a word that distinguishes
between them. 
A related research topic is the problem of finding a minimal \dfa that
 distinguishes between two words by accepting one and rejecting the other.
It was studied by~Blumer et
al. in \cite{blumer85:_small_autom_recog_subwor_of_text}, and recently 
Demaine et al. in \cite{demaine11:_remar_separ_words} 
reviewed several attempts to solve the problem and presented new results. 

In the present paper we do not separate two words by a language,
instead, we distinguish between non-equivalent quotients of the same
language. We use many powerful tools such as 
language quotients, atoms, and universal witnesses, 
that hide proof complexity,
helping us to produce a presentation easier to follow.
We introduce the notation in Section~\ref{snotation},  define and
prove general properties of the distinguishability operation in
Section~\ref{sgenres}, and prove some state complexity results in
Section~\ref{ssc}.  
In Section~\ref{sec:minDist}, we analyze the set of minimal words with
respect to quasi-lexicographical order that distinguishes different
quotients of a regular language. 
In Section~\ref{sec:recoverL}, we present an algorithm that can be used
as a positive learning procedure for language $L$ if $\distmin{L}$ is known.
In Section~\ref{sec:sufclosuregeneral}, we present a class of operands 
using closure operations and Boolean operations, that have a 
fixed point under iteration. 
In Section~\ref{sec:otherdis} we define other distinguishability operations and study their properties.  
The conclusion, together with open problems and future work, are included 
in Section~\ref{sconc}.

\section{Notation and Definitions}
\label{snotation}
For a set $T$, its cardinality is denoted by $|T|$. 
An \emph{alphabet} $\Sigma$ is a finite non-empty set, and the free monoid
generated by $\Sigma$ is $\Sigma^\star$. 
A word $w$ is an element of $\Sigma^\star$ and a \emph{language} is a subset of $\Sigma^\star$.  
The \emph{complement} of a language $L$ is $\comp{L}=\Sigma^\star\setminus L$.
The \emph{length} of a word $w\in \Sigma^\star$, $w=a_1a_2\ldots a_n$,
$a_i\in \Sigma$, $1\leq i\leq n$, with $n\in\N$ is $|w|=n$.  
The \emph{empty word} is $\varepsilon$, and $|\varepsilon|=0$. If $w=uxv$ for some $u,v,x\in \Sigma^\star$ then $u$ is a \emph{prefix} of $w$, $x$ is a \emph{factor} (or \emph{infix}) of $w$ and $v$ a \emph{suffix} of $w$.
Consider an order over $\Sigma$. 
In $\Sigma^\star$, we define the \emph{quasi-lexicographical} order as: $w\preceq w'$
if $|w|<|w'|$ or $|w|=|w'|$ and $w$ lexicographically precedes $w'$. 
The reverse $w^R$ of a
word $w\in \Sigma^\star$ is defined as follows: $\varepsilon^R=\varepsilon$, and $(wa)^R=aw^R$, for $a\in \Sigma$. 
The \emph{reverse} of a language $L$ is denoted by $L^R$ and defined as $L^R =\{w^R\mid w\in L\}$.

A \emph{deterministic finite automaton} (\dfa) is a quintuple
$\mathcal{A}=\tuple{Q,\Sigma, q_0,\delta,F}$, where $Q$ is  a finite
non-empty set, the set of states, $\Sigma$ is the  alphabet, $q_0\in
Q$ is the initial state, $F\subseteq Q$ is the set of final states,
and $\delta:Q\times \Sigma\longrightarrow Q$ is the transition
function. This function defines for each symbol of
the alphabet a transformation of the set $Q$ of states (i.e. a map from $Q$ to $Q$). The
\emph{transition semigroup} of a \dfa $\mathcal{A}$, 
\cite{bell14:_symmet_group_and_quotien_compl},
 is the semigroup of transformations of $Q$ generated 
by the transformations induced by the symbols of $\Sigma$.

A \emph{reduced} \dfa is a \dfa with all states reachable from the  
initial state (accessible), and all states can reach a final state 
(useful), except at most one that is a \emph{sink} state or
\emph{dead} state, i.e., a state where all output transitions are self
loops. 

The transition function $\delta$ can be extended to $\delta:Q\times
\Sigma^\star\longrightarrow Q$ by $\delta(q,\varepsilon)=q$, and
$\delta(q,wa)=\delta(\delta(q,w),a)$. 
The language recognized by a \dfa $\mathcal{A}$ is
$\lang{\mathcal{A}}=\{w\in \Sigma^\star\mid \delta(q_0,w)\in F\}$.  
We denote by 
$L_q$ and $R_q$ the \emph{left} and \emph{right} languages of $q$, respectively,
i.e.,  
$L_q = \Set{w \mid  \delta(q_0,w) = q}$, and 
$R_q = \Set{w \mid  \delta(q,w)\in F}$. 

The minimal word in quasi-lexicographical order that reaches state $q\in Q$
is $x_{\cal A}(q)$; the word $x_{\cal A}(q)$ is also the minimal element of $L_q$.

A \emph{regular language} is a language recognized by a \dfa.  
A regular language $L$ 
induces on $\Sigma^\star$ the Myhill-Nerode equivalence relation:
 $x\mneq{L} y$ if, for all $w\in
\Sigma^\star$, we have that $xw\in L$ if and only if $yw\in L$. 
If $\mathcal{A}=\tuple{Q,\Sigma,\delta,q_0,F}$ is a \dfa recognizing
the language $L$ and  $R_q=R_p$, then we say that $p$ and $q$ are equivalent, 
and write $p\mneq{\mathcal{A}} q$.   
A \dfa is \emph{minimal} if it has no equivalent states.
The \emph{left quotient}, or simply \emph{quotient}, of a regular
language $L$ by a word $w$ is the language $w^{-1}L=\{x\mid wx \in
L\}$. A quotient 
corresponds to an equivalence class of $\mneq{L}$, i.e. two words are equivalent if and only if their quotients are the same. 
If a language $L$ is regular, the number of distinct left quotients is
finite, and it is exactly the number of states in the \emph{minimal} \dfa
recognizing $L$. This number is called the \emph{state complexity} of $L$,
 and is denoted by $sc(L)$. 
In a minimal \dfa, for each $q\in Q$, $R_q$ is exactly a quotient. 
If some quotient of a language $L$ is $\emptyset$, this means that the
minimal \dfa of $L$ has a dead state. 

A \emph{nondeterministic finite automata}  (\nfa) is a quintuple
$\mathcal{N}=\tuple{Q,\Sigma, I,\delta,F}$, where $Q$, $\Sigma$, and
$F$ are the same as in the \dfa definition, $I\subseteq Q$ is the set of
initial states, and $\delta:Q\times \Sigma\longrightarrow 2^Q$ is the
transition function.  
The transition function can also be extended to subsets of $Q$ instead
of states, and to words instead of symbols of $\Sigma$. 
The language recognized by  an \nfa $\mathcal{N}$ is
$\lang{\mathcal{N}}=\{w\mid \delta(I,w)\cap F\not= \emptyset\}$.  
It is obvious that a \dfa is also an \nfa. 
Any \nfa can be converted in an equivalent \dfa by the well known
\emph{subset construction}. Given an \nfa $\mathcal{N}$ for $L$, an \nfa $\mathcal{N}^R$ for $L^R$ is obtained by interchanging the sets of final and initial states of $\mathcal{N}$ and reversing all transitions between states.

More notation and definitions related to formal languages can be consulted in~\cite{sakarovitch09:_elemen_of_autom_theor,yu97:_handb_formal_languag}.

\section{The (Left) Distinguishability Operation}
\label{sgenres}
Let $L$ be a regular language. 
For every pair of words, $x,y\in \Sigma^\star$, with $x\not\mneq{L} y$,
there exists at least one word $w$ such that either $xw \in L$
or $ yw \in L$. 
Let  ${\cal A}=\tuple{Q,\Sigma,\delta,q_0,F}$ be a \dfa such that
$L=\lang{\cal A}$.  
If two states $p,q\in Q$, $p\not\mneq{A} q$, then there exists at least one
word $w$ such that $\delta(p,w) \in F \notequiv\delta(q,w)
\in F$.  
We say that $w$ distinguishes between the words $x$ and $y$, in the
first case, and the states $p$ and $q$, in the second case. Given $x,y\in\Sigma^\star$, the language that \emph{distinguishes} $x$ from
$y$ w.r.t. $L$ is 
\begin{equation}
\label{eq:dw}
\disw{L}{x,y}=\Set{w\;|\; xw\in L\notequiv yw\in L}.
  \end{equation}

Naturally, we define the 
\emph{left distinguishability language} (or simply, distinguishability language) of $L$ by
\begin{equation}
\label{edisdefL}
\dis{L}=\{w\mid \exists x,y\in \Sigma^\star \ (xw \in L\ \wedge \ yw
\notin L)\}.  
\end{equation}

It is immediate that $\dis{L}=\displaystyle\bigcup_{x,y\in
  \Sigma^\star} \disw{L}{x,y}$.
%
%
In the same way, for the \dfa $\mathcal{A}$, we define $\disw{L}{p,q}$
for $p,q\in Q$, and 
\begin{equation}
\label{edisdefA}
\dis{\mathcal{A}}=\{w\mid \exists p,q\in Q \  (\delta(p,w)\in F\ \wedge \ 
                          \delta(q,w)\notin F) \}.
\end{equation}

\begin{lemma}
 Let $\mathcal{A}_1,\mathcal{A}_2$ be two reduced \dfas such that 
$\lang{\mathcal{A}_1}=\lang{\mathcal{A}_2}=L$. 
Then $\dis{\mathcal{A}_1}=\dis{\mathcal{A}_2}=\dis{L}$.
\end{lemma}
\begin{proof}
Let $\mathcal{A}_1=\tuple{Q,\Sigma, q_0,\delta,F}$ and
$L=\lang{\mathcal{A}_1}=\lang{\mathcal{A}_2}$.  
It is enough to prove that $\dis{L}=\dis{\mathcal{A}_1}$.
If $w\in \dis{L}$, then we have two words $x,y\in \Sigma^\star$ such that
$xw\in L$ and $yw\notin L$. 
Let $p=\delta(q_0,x)$ and $q=\delta(q_0,y)$. 
Then, $\delta(q_0,xw)\in F$ and $\delta(q_0,yw)\notin F$, so $w\in
\dis{\mathcal{A}_1}$. 
If  $w\in \dis{\mathcal{A}_1}$, then there exist $p, q\in Q$ such that 
$\delta(p,w)\in F$ and  $\delta(q,w)\notin F$; as $\mathcal{A}_1$ is reduced, 
there must exist $x,y$ with $\delta(q_0,x)=p$ and $\delta(q_0,y)=q $, 
therefore $\delta(q_0,xw)= \delta(p,w)\in F$ and $\delta(q_0,yw)=\delta(q,w)\notin F$,
 hence, $xw\in L$ and $yw\notin L$, i.e., we conclude $w\in \dis{L}$.
\end{proof}

This shows that the operator $\diso$ is independent of the automata
we choose to represent the language. In what follows, we present some characterization results for the
distinguishability operation, and show that iterating the operation
always leads to a fixed point. 

The distinguishability operation can be expressed directly by means of
the language quotients, as it can be seen in the following result. 
\begin{theorem}
\label{theo:dist-quot}
Let $L$ be a regular language with $\Set{w^{-1}L\;|\;w\in
  \kleene{\Sigma}}$ its set of (left) quotients. Then we have the equality
\begin{equation*}
 \dis{L} = \bigcup_{x\in\kleene{\Sigma}}x^{-1}L \setminus
 \bigcap_{x\in\kleene{\Sigma}}x^{-1}L. 
\end{equation*}
\end{theorem}
\begin{proof}
Let $w\in \dis{L}$, i.e.,
 $\exists x,y\in\Sigma^\star, xw\in L
                          \AND yw\notin L$. 
Then $w\in x^{-1}L \AND w\notin y^{-1}L$, therefore  
${\displaystyle \dis{L} \subseteq \bigcup_{x\in\kleene{\Sigma}}x^{-1}L
  \setminus \bigcap_{x\in\kleene{\Sigma}}x^{-1}L}$. 

Let ${\displaystyle w\in \bigcup_{x\in\kleene{\Sigma}}x^{-1}L
  \setminus \bigcap_{x\in\kleene{\Sigma}}x^{-1}L}$, then let $x$ and
$y$ be such that $w\in x^{-1}L$ and $w\notin y^{-1}L$. Thus, $w\in
\dis{L}$. Hence, we conclude
$$\dis{L}=\bigcup_{x\in\kleene{\Sigma}}x^{-1}L \setminus
\bigcap_{x\in\kleene{\Sigma}}x^{-1}L.$$ 
\end{proof}
\begin{corollary}
\label{cor;diffxyquot}
For a regular language $L$ and $x,y\in \Sigma^\star$, $\disw{L}{x,y}$ is 
the symmetric difference of the correspondent quotients, 
$\disw{L}{x,y}=(x^{-1}L)\Delta (y^{-1}L).$
\end{corollary}

To help the reader better understand  how these languages look like, next we present 
some examples.

\begin{example}
\label{ex:singleton}
 If the language $L$ has only one quotient, i.e., $L=\emptyset$ or
 $L=\kleene{\Sigma}$, then $\dis{L}=\emptyset$, as there are no
 different quotients to distinguish. 
\end{example}

\begin{example}
\label{ex:epsilon}
If $\dis{L}=\{\varepsilon\}$, then we can only distinguish between final 
and non-final states, thus the minimal \dfa of $L$ has exactly two states 
corresponding to its  two quotients.
\end{example}

\begin{example}
\label{ex:symmerysg}
In this example we consider a family of languages $L_n$ for which $\dis{L_n}=\Sigma^\star$.
Let $L_n=\lang{A_n}$ for $3\leq n$,  with
$A_n=\tuple{Q_n,\{0,1\},\delta_n,0,\{0\}}$,  
where $Q_n=\Set{0,\ldots,n-1}$, 
$\delta_n(i,0)=i+1\mod{n-1}$, for $0\leq i\leq n-1$ and 
$\delta_n(1,1)=0$, $\delta_n(0,1)=1$,  $\delta_n(i,1)=i$, for $2\leq
i\leq n-1$.  
In Figure~\ref{fSigmaS}, we present~$A_5$.   
Both symbols of the alphabet induce permutations on $Q_n$: $1$ induces
a transposition (2-cycle), and  $0$, an $n$ cyclic permutation. It
follows that, the transition semigroup of $A_n$  
is the symmetric group $S_n$ of degree $n$, i.e., the set of all
permutations of $Q_n$.  
We always have  $\varepsilon\in\dis{L}$, and  
 every $w\in \Sigma^+$ induces a permutation on the states, 
for  $0\leq i,j\leq n-1$, $\delta(i,w)=i_w$ and  $\delta(j,w)=j_w$,
with $i_w\neq j_w$.  
Then, there must exist at least a pair $(i,j)$, such that $w\in R_i$ and
$w\notin R_j$, i.e., 
 $w\in  \dis{L_n}$,  thus $\dis{L_n}=\Sigma^\star$. 
Bell et al.~\cite{bell14:_symmet_group_and_quotien_compl}
studied those families of automata,  
and in particular, proved that they are uniformly minimal, i.e.,
minimal for every non-trivial choice of final
states~\cite{restivo12:_graph_theor_approac_to_autom_minim}.  
\end{example}

\begin{figure}[htb]
\begin{center}

\hfill\begin{tikzpicture}[>=stealth, shorten >=1pt, auto, node distance=1.5cm,initial text={}]

\node[state, accepting, initial, inner sep=1pt, minimum size=7pt] (S0){};
\node[state, inner sep=1pt, minimum size=15pt][below of=S0, minimum size=7pt] (S1){};
\node[state, inner sep=1pt, minimum size=15pt][right of=S1,minimum size=7pt] (S2){};
\node[state, inner sep=1pt, minimum size=15pt][above right of=S2, yshift = -.4cm, minimum size=7pt] (S3){};
\node[state, inner sep=1pt, minimum size=15pt][right of=S0, minimum size=7pt] (S4){};
\path[->](S0) edge [bend left]node {$0,1$} (S1)
         (S1) edge node [swap]{$0$} (S2)
		      edge [bend left] node {$1$} (S0)
		 (S2) edge node [swap]{$0$} (S3)
		      edge [loop below] node {$1$} ()
		 (S3) edge node {$0$} (S4)
		      edge [loop right] node {$1$} ()
		 (S4) edge [loop above] node {$1$} ()
		      edge node [swap] {$0$} (S0);
\end{tikzpicture}\quad
\begin{tikzpicture}[>=stealth, shorten >=1pt, auto, node distance=2cm,initial text={}]

\node[state, accepting, initial, inner sep=1pt, minimum size=7pt] (S0){};
\path[->]
		 (S0) edge [loop below] node {$0,1$} ();
\end{tikzpicture}\hfill~
\caption{Automaton  $A_5$ (left) and its distinguishability language,  $\Sigma^\star$ (right).
  } 
\label{fSigmaS} 
\end{center}
\end{figure}
\begin{figure}[h!]
\begin{center}
\begin{tikzpicture}[>=stealth, shorten >=1pt, auto, node distance=1.5cm,initial text={}]

\node[state, accepting, initial, inner sep=1pt, minimum size=7pt] (S0){};
\node[state, accepting, inner sep=1pt, minimum size=7pt] [right of=S0](S1){};
\node[state, accepting, inner sep=1pt, minimum size=7pt] [above of=S1](S2){};
\node[state, accepting, inner sep=1pt, minimum size=7pt] [right of=S2](S3){};
\node[state, accepting, inner sep=1pt, minimum size=7pt] [below of=S3](S4){};
\node[state, accepting, inner sep=1pt, minimum size=7pt] [right of=S4](S5){};
\node[state, accepting, inner sep=1pt, minimum size=7pt] [above of=S5, yshift=20pt](S6){};
\node[state, inner sep=1pt, xshift=1cm, minimum size=7pt] [right of=S5](S7){};
\node[state, inner sep=1pt, minimum size=7pt] [above of=S7](S8){};
\node[state, accepting, inner sep=1pt, minimum size=7pt] [right of=S7](S9){};
\node[state, inner sep=1pt, minimum size=7pt] [above of=S9](S10){};
\path[->] (S0) edge [loop above] node {$1$}()
               edge node {$0$}(S1)
		  (S1) edge [loop below] node {$0$}()
		       edge node {$1$}(S2)
		  (S2) edge node {$1$}(S3)
		       edge [bend left] node [pos=.2] {$0$}(S6)
		  (S3) edge node {$0$}(S1)
			   edge node {$1$}(S4)
		  (S4) edge node {$1$}(S1)
		       edge node {$0$}(S5)
		  (S5) edge [loop below] node {$0$}()
		       edge node {$1$}(S6)
		  (S6) edge node {$0$}(S7)
		       edge node [pos=.2]{$1$}(S8)
		  (S7) edge [loop left] node {$0$}()
		       edge node {$1$}(S8)
		  (S8) edge [out=120, in=30] node[swap, pos=.2] {$0$}(S6)
		       edge node [pos=.2]{$1$}(S9)
		  (S9) edge node [swap]{$0,1$}(S10)
		  (S10) edge [out=140,in=50] node [pos=.1,swap] {$1$}(S6)
		        edge [out=260,in=300] node {$0$} (S5);
\end{tikzpicture}
\end{center}
\caption{Example of an automaton with $\lang{\mathcal{A}}\neq
\dis{\lang{\mathcal{A}}}$.} 
\label{fdiff} 
\end{figure}

\begin{example}
\label{ex:ts45}
Consider the automaton ${\cal A}$ in Figure~\ref{fdiff}.\ 
We have that $\lang{{\cal A}}\neq \dis{\lang{\mathcal{A}}}$,  
but $\dis{\dis{\lang{{\cal A}}}}= \dis{\lang{\mathcal{A}}}$.  
The minimal automaton for $\dis{\lang{\mathcal{A}}}$ is presented in Figure~\ref{ffixp}.
\end{example}

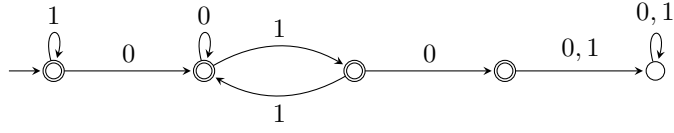
\begin{figure}
\begin{center}
\begin{tikzpicture}[>=stealth, shorten >=1pt, auto, node distance=2cm,initial text={}]

\node[state, accepting, initial, inner sep=1pt, minimum size=7pt] (S0){};
\node[state, accepting, inner sep=1pt, minimum size=7pt] [right of=S0](S1){};
\node[state, accepting, inner sep=1pt, minimum size=7pt] [right of=S1](S2){};
\node[state, accepting, inner sep=1pt, minimum size=7pt] [right of=S2](S3){};
\node[state, inner sep=1pt, minimum size=7pt] [right of=S3](S4){};
\path[->] (S0) edge node {$0$} (S1)
	           edge [loop above] node {$1$}()
		  (S1) edge [loop above] node {$0$}()
		       edge [bend left] node {$1$}(S2)
		  (S2) edge [bend left] node{$1$}(S1)
		       edge node {$0$} (S3)
		  (S3) edge node {$0,1$} (S4)
		  (S4) edge [loop above] node {$0,1$}();
\end{tikzpicture}
\end{center}
\caption{Example of automaton where $\dis{\lang{\mathcal{A}_1}}=\lang{\mathcal{A}_1}$, i.e. distinguishability  
language is the same as the language of the words it can distinguish.}
\label{ffixp} 
\end{figure}

\begin{example}
\label{ex:suff}
Considering the language $L=\kleene{((0+1)(0+1))} (\varepsilon + 1)$, 
in Figure~\ref{f1} one can find, from left to right, a \dfa that accepts $L$, 
one that accepts  $\dis{L}=\kleene{(0 +  \kleene{1}10)}$, 
and one  for  $\disn{L}{n}=\varepsilon$, for $n\geq 2$.
\end{example}

From the last example, we can see that the language $\dis{L}$ contains
the word $0110$,  
but also the words $110$, $10$,  and $0$, which are all suffixes of  $0110$. 
This observation suggests that $\dis{L}$ is suffix closed, which is proved 
in the following theorem.

\begin{figure}[h!]
\hfill\begin{tikzpicture}[>=stealth, shorten >=1pt, auto, node distance=2.3cm,initial text={}]
\node[state, accepting, initial, inner sep=1pt, minimum size=7pt] (S0){};
\node[state, inner sep=1pt, minimum size=15pt, minimum size=7pt, yshift=-.5cm][above right of=S0] (S1){};
\node[state, accepting, inner sep=1pt, minimum size=15pt, minimum size=7pt, yshift=.5cm][below right of=S0] (S2){};
\path[->](S0) edge [bend left] node{$0$}(S1)
              edge node{$1$}(S2)
		 (S1) edge node{$0,1$}(S0)
         (S2) edge [bend left] node{$0,1$}(S0);
\end{tikzpicture}
\hfill\begin{tikzpicture}[>=stealth, shorten >=1pt, auto, node distance=2cm,initial text={}]
\node[state, accepting, initial, inner sep=1pt, minimum size=7pt] (S0){};
\node[state, inner sep=1pt, minimum size=15pt, minimum size=7pt][right of=S0] (S1){};
\path[->](S0) edge [bend left] node{$1$}(S1)
              edge [loop above] node{$0$}()
		 (S1) edge [bend left] node{$0$}(S0)
		      edge [loop above] node{$1$}();
\end{tikzpicture}
\hfill\begin{tikzpicture}[>=stealth, shorten >=1pt, auto, node distance=1.5cm,initial text={}]
\node[state, accepting, initial, inner sep=1pt, minimum size=7pt] (S0){};
\node[state, inner sep=1pt, minimum size=15pt, minimum size=7pt][right of=S0] (S1){};
\path[->] (S0) edge node {$0,1$} (S1)
          (S1) edge [loop below] node {$0,1$}();
\end{tikzpicture}\hfill~
\caption{Automata for the languages $L$, $\dis{L}$, and $\disn{L}{n}$,
  $n\geq 2$.} 
\label{f1}
\end{figure}

\begin{theorem}
\label{tsuffcl}
    If $L$ is a regular language, then the language $\dis{L}$ is suffix
    closed, i.e.,  
    $$(\forall w\in \dis{L})(\forall x,y\in \Sigma^\star)(w=xy\implies
    y\in \dis{L}).$$ 
\end{theorem}
\begin{proof}
Let $w\in \dis{L}$, i.e., there exist $x,y\in \Sigma^\star$ such that  
$xw\in L$ and $yw\notin L$. 
If $v$ is a suffix of $w$, i.e., $w=uv$, for an $u\in \Sigma^\star$, 
then we can write $xuv\in L$ and $yuv\notin L$, which means that $v\in \dis{L}$.
\end{proof}

Using Theorem~\ref{theo:dist-quot}, 
if $w\in \dis{L}$, then 
$w$ is a suffix of a word in $L$, and a suffix of the complement of $L$,
because 
$$\bigcup_{x\in\kleene{\Sigma}}x^{-1}L \setminus
\bigcap_{x\in\kleene{\Sigma}}x^{-1}L=\bigcup_{x\in\kleene{\Sigma}}x^{-1}L
\bigcap \bigcup_{x\in\kleene{\Sigma}}x^{-1}\comp{L}.$$ 

 Accordingly,
$\dis{L}\subseteq \suff{L}\cap \suff{\comp{L}}$, where $\suff{L}$ is
the language of all suffixes of $L$. 
If $w \in \suff{L}\cap \suff{\comp{L}}$, then we can find $x$ and $y$ such that
$xw\in L$ and $yw\in \comp{L}$, thus $w\in \dis{L}$.
Therefore, we just found a new way to express the distinguishability
language of $L$: 

\begin{theorem}
\label{theo:dss} If $L$ is a regular language, then
\begin{equation}
\label{eq:dss}  
\dis{L}=\suff{L}\cap\suff{\comp{L}}.
\end{equation}
\end{theorem}

Because $\dis{L}$ is suffix closed,  $\dis{L}=\suff{\dis{L}}$, hence
$\suff{\dis{L}}\subseteq \dis{L}\subseteq \suff{L}$ and
$\disn{L}{2}\subseteq \dis{L}\subseteq \suff{L}$. 
In general, we have for every $n\geq 1$, the following inclusion
\begin{equation}
\label{eqchain}
\disn{L}{n+1}\subseteq \disn{L}{n}. 
\end{equation}

Consequently, we may ask if this hierarchy is infinite or not, in other words, 
we may ask if for any language $L$,
there exists $n\geq 0$ such that $\disn{L}{n+1}= \disn{L}{n}$. 

\begin{example}
\label{ex:d3}
Consider the language $L=\lang{\cal A}$, where ${\cal A}$ is given 
in Figure~\ref{f2}, on the left. For the language $L$, we have that $L\neq \dis{L}$ and 
 $\dis{L}\neq\disn{L}{2}=\disn{L}{n}$, for $n\geq 2$. 
The minimal automaton for  $\disn{L}{2}$ is depicted on the right. 
The minimal automaton for $\dis{L}$ has $7$ states.
\end{example}

\begin{figure}[h!]
\begin{center}
\hfill\begin{tikzpicture}[>=stealth, shorten >=1pt, auto, node distance=1.9cm,initial text={}]

\node[state, initial, inner sep=1pt, minimum size=7pt] (S0){};
\node[state, inner sep=1pt, minimum size=7pt][above of=S0,yshift=-6pt] (S1){};
\node[state, accepting, inner sep=1pt, minimum size=7pt][right of=S1] (S2){};
\node[state, inner sep=1pt, minimum size=7pt][right of=S2] (S3){};
\node[state, accepting, inner sep=1pt, minimum size=7pt,yshift=9pt][below of=S2] (S4){};
\node[state, accepting, inner sep=1pt, minimum size=7pt][below of=S0, yshift=-0pt] (S5){};
\node[state, inner sep=1pt, minimum size=7pt][right of=S5] (S6){};
\node[state, inner sep=1pt, minimum size=7pt][right of=S6] (S7){};
\node[state, accepting, inner sep=1pt, minimum size=7pt, yshift=10pt][above of=S7] (S8){};
\path[->] (S0) edge node {$0$}(S1)
               edge node {$1$}(S5)
		  (S1) edge [loop above] node {$1$}()
		       edge node {$0$}(S2)
		  (S2) edge [loop above] node {$0$}()
		       edge node {$1$}(S3)
		  (S3) edge node [pos=.2] {$0,1$}(S4)
		  (S4) edge [loop above] node {$0$}()
		       edge node {$1$}(S1)
		  (S5) edge node {$0,1$}(S6)
		  (S6) edge [loop above] node {$0$}()
		       edge node {$1$}(S7)
		  (S7) edge [loop right] node {$1$}()
		       edge node {$0$}(S8)
		  (S8) edge [loop right] node {$0$}()
		       edge node[pos=.3,swap] {$1$}(S5);
\end{tikzpicture}\hfill
\begin{tikzpicture}[>=stealth, shorten >=1pt, auto, node distance=1.9cm,initial text={}]

\node[state, initial, accepting, inner sep=1pt, minimum size=7pt] (S0){};
\node[state, accepting,inner sep=1pt, minimum size=7pt][right of=S0] (S1){};
\node[state, accepting,inner sep=1pt, minimum size=7pt][below right of=S0] (S2){};
\node[state, inner sep=1pt, minimum size=7pt][right of=S2] (S3){};
\path[->] (S0) edge [loop above] node {$0$}()
			   edge [bend left] node {$1$}(S1)
		  (S1) edge [bend left] node {$0$}(S0)
		       edge node {$1$}(S2)
		  (S2) edge node {$0$}(S0)
		       edge node {$1$}(S3)
		  (S3) edge [loop below] node {$0,1$}();

\end{tikzpicture}\hfill~
\end{center}
\caption{Example of a language $L$ with
$\dis{L}\neq\disn{L}{2}=\disn{L}{n}$, for $n\geq 3$. On the left a \dfa for $L$ and on the right a \dfa for $\disn{L}{2}$.} 
\label{f2}
\end{figure}
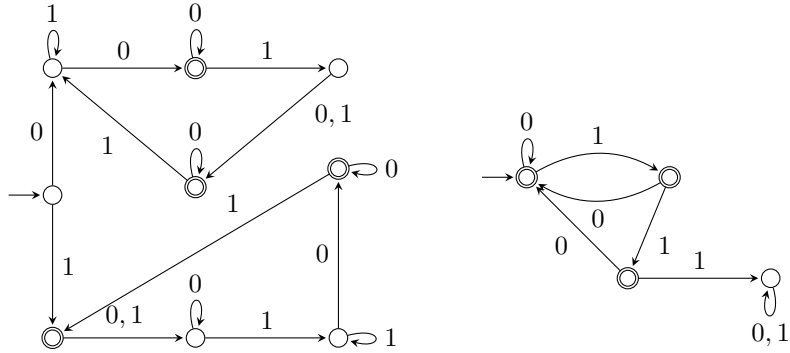

The following lemma will be useful for the rest of the section.

\begin{lemma}
\label{lsufcap}
 If $L,M\subseteq \Sigma^\star$ are suffix-closed languages, then
$\suff{L}\cap\suff{M}=\suff{L\cap M}$, and 
$\suff{L}\cup\suff{M}=\suff{L\cup M}$.
\end{lemma}
\begin{proof}
It is obvious that the equality holds for reunion, and  
the inclusion $\suff{L\cap M}\subseteq \suff{L}\cap\suff{M}$, for intersection is true.
If $w\in \suff{L}\cap\suff{M}$, then
there exist $x,y\in \Sigma^\star$ such that $xw\in L$ and $yw\in M$. 
Because $L$ and $M$ are suffix closed, then
$w\in L\cap M\subseteq \suff{L\cap M}.$
\end{proof}

In the following result, we prove that the iteration of $\diso$
operations always reaches a fixed point. 

\begin{theorem}
\label{theo:fixpoint}
Let $L\subseteq \kleene{\Sigma}$ be a regular language. 
Then we have that  
  $\disn{L}{3}=\disn{L}{2}$.
\end{theorem}

\begin{proof}
We have the following equalities:
\begin{eqnarray}
\label{eq:d2}
  \disn{L}{2}&=&\dis{\dis{L}}=\suff{\dis{L}}\cap
  \suff{\comp{\dis{L}}}= \dis{L}\cap 
  \suff{\comp{\dis{L}}}.
   \end{eqnarray}
  
Now, computing the next iteration of $\diso$, we get
 $\disn{L}{3}=\suff{\disn{L}{2}}\cap  \suff{\comp{\disn{L}{2}}}$.

Using~(\ref{eq:d2}) and Lemma~\ref{lsufcap},  we obtain the equalities
\begin{eqnarray*}
\suff{\comp{\disn{L}{2}}}=\suff{\comp{\dis{L}\cap\suff{\comp{\dis{L}}}}}&=&
  \suff{\comp{\dis{L}}}\cup  \suff{\comp{\suff{\comp{\dis{L}}}}}.
\end{eqnarray*}
Because $\disn{L}{2}$ is a suffix-closed language, it follows that
\begin{eqnarray*}
\disn{L}{3} & = & \disn{L}{2}\cap (\suff{\comp{\dis{L}}}\cup \suff{\comp{\suff{\comp{\dis{L}}}}}) \\
            & = & (\disn{L}{2}\cap \suff{\comp{\dis{L}}})\cup (\disn{L}{2}\cap \suff{\comp{\suff{\comp{\dis{L}}}}}))\\
            & = & \disn{L}{2}\cup (\disn{L}{2}\cap \suff{\comp{\suff{\comp{\dis{L}}}}}))=\disn{L}{2}.
\end{eqnarray*}
\end{proof}

The following results give some characterization for the languages
that are fixed points for $\diso$. 

\begin{lemma}
\label{lem:qempty}
Given a regular language $L$, if $L$ has $\emptyset$ as a quotient then $\dis{L}=\suff{L}$.
\end{lemma}

\begin{proof}
Because $z^{-1}L=\emptyset$, for some word $z$, we have that
$\kleene{\Sigma}= \suff{\comp{L}}$. 
\end{proof}
This lemma makes the following result immediate.
\begin{theorem}
\label{theo:ldirectdead}
 If $L$ is a suffix closed regular language with $\emptyset$ as one of the
 quotients,  then $L$ is a fixed point for $\diso$, i.e., $\dis{L}=L$.
\end{theorem}

\begin{corollary}
\label{cor:emptyquotient}
Let $L$ be a regular language. If $\dis{L}$ has $\emptyset$ as a quotient, then $\disn{L}{2}=\dis{L}$.
 \end{corollary}

Note that suffix-closeness of $L$ is not sufficient to ensure that $L$ has $\emptyset$ as quotient.
 For that, it is enough to consider the language given by $\kleene{0}+\kleene{0}1\kleene{(1+0\kleene{0}1)}$.
 However, if $L$ is a $\diso$ fixed point, the implication yields.

\begin{theorem}
\label{ldead}
    Let $L$ be a regular language.
If $\dis{L}=L$, then $L$ has $\emptyset$ as a quotient. 
\end{theorem}
\begin{proof}
Let $L$ be a regular language that is fixed point for $\diso$, thus $L$ is suffix closed and 
\begin{equation}\label{fix-point}
(\forall w\in L)(\exists u\in\Sigma^\star) (uw\notin L).
\end{equation}
Assume that $L$ does not have $\emptyset$ as quotient, i.e., 
\begin{equation}\label{ecological}
(\forall w\in\Sigma^\star)(\exists v\in \Sigma^\star)(wv\in L).
\end{equation}
Let $w\notin L$ ($\Sigma^\star$ is not a fixed point for $\diso$). 
Thus by~\eqref{ecological} there  exists a $v_0\in \Sigma^\star$ such that $wv_0\in L$
 and because $L$ is suffix closed, it follows that $v_0\in L$. 
Using \eqref{fix-point} there exists $u_0$ such that $u_0wv_0\notin L$. 
Using the same reasoning, we can find $u_1\in\Sigma^\star$ and $v_1\in L$ such that 
$$u_0wv_0v_1\in L
\text{ and } u_1u_0wv_0v_1\notin L.
$$
The word $wv_0v_1$ distinguishes $u_0$ from $u_1u_0$, thus these words cannot belong to the same quotient. Suppose that we have iterated $n$ times this process having 
$$u_{n-1}\cdots u_0wv_0\cdots 
v_n\in L \text{ and } u_n\cdots u_0wv_0\cdots v_n\notin L,$$
with all $u_i\cdots u_0$ belonging to 
distinct quotients. We can apply this process one more time, obtaining 
$$u_n\cdots u_0wv_0\cdots v_{n+1}\in L \text{ and } u_{n+1}\cdots u_0wv_0\cdots v_{n+1}\notin L.$$ 
It is easy to see that the word $wv_0\cdots v_{n+1}$ distinguishes $u_{n+1}\cdots u_{0}$ from 
any of the previous words $u_i\cdots u_0$ (with $i\leq n$) because $u_i\cdots u_0wv_0\cdots 
v_{n+1}$ is a suffix of $u_n\cdots u_0wv_0\cdots v_{n+1}\in L$.
Thus, the number of $L$ quotients cannot be finite, a contradiction.
\end{proof}

By contraposition over the last result, we get that  a language $L$
with all its quotients non-empty cannot be a fixed point for $\diso$. 
 We know that, if a language $L$ is such
that $\kleene{\Sigma}=\suff{\comp{L}}$, then $\emptyset$ must be one of
the quotients of $L$. 
In Examples~\ref{ex:symmerysg}--\ref{ex:suff}
and Example~\ref{ex:d3}, we have  languages $L$ with all quotients non-empty and  $L\neq\dis{L}$. Considering Theorem~\ref{ldead} and Theorem~\ref{fix-point}, given any regular language $L$ we can iterate $\diso$ at most two times to obtain a language that has $\empty$ as a quotient.

For a finite language $L$, 
it follows from Lemma~\ref{lem:qempty}
that the
distinguishability language of $L$ coincides with the set of all suffixes of $L$,
therefore, $\dis{L}$ is a fixed point of the $\diso$ operator.

\begin{corollary}
\label{cor:finite}
  If $L$ is a finite language, then $\dis{L}=\suff{L}$.
\end{corollary}


The minimal \dfa that represents the set of suffixes of a finite
language $L$ is called the \emph{suffix automaton}, and several
optimized algorithms for its construction were studied in the
literature. Thus, we can use an algorithm for building the suffix
automaton  in order to obtain $\dis{L}$. Recently, Mohri et
al.~\cite{mohri09:_gener_suffix_autom_const_algor} gave new upper
bounds on the number of states of the suffix automaton as a function
of the size of the minimal \dfa of $L$, as well as other measures of
$L$. In Section~\ref{ssc}, we study the state complexity of
$\dis{L}$ as a function of the state complexity of $L$, for any general regular language $L$. In Section~\ref{sec:sufclosuregeneral}, we generalize the results on the
characterization of the distinguishability operation, defining more operations on regular languages. 

\section{State Complexity}
\label{ssc}
By Theorem~\ref{theo:dss}, we know that $\dis{L}$ can be obtained  
using the suffix operator, complement and intersection,
therefore, it is a result of combining three operations, two unary and one binary.
We would like to estimate the state complexity of the $\diso$ operation and 
check if the upper bound is tight.
We recall that the state complexity of an operation is the worst-case state complexity of a language resulting from that operation, as a function of the state complexities of the operands.
 
The following theorem shows the construction for $\dis{L}$, in case $L$ is recognized by a \dfa.
 
\begin{theorem}
\label{theo:algodis}
  Let ${\cal A}=(Q,\Sigma,\delta,i,F)$ be a  reduced \dfa recognizing a language $L$.\ 
Then ${\cal A}_d=(Q_d,\Sigma,\delta_d,Q,F_d)$ is a \dfa that accepts $\dis{L}$, 
where 
\begin{itemize}
 \item $Q_d\subseteq 2^{Q}$, 
 \item for $a\in \Sigma$ and $S\subseteq Q$ and $|S|>1$, 
$\delta_d(S,a)=\{\delta(q,a)\mid q\in S\}$,
 \item $F_d=\{ S\mid S\cap F\not= \emptyset \text{ and }  S\cap (Q\setminus F)\not= \emptyset\}$.
\end{itemize}
\end{theorem}
\begin{proof}
Considering that $\dis{L}=\suff{L}\cap\suff{\comp{L}}$, we can use
the usual subset construction for $\suff{L}$ to build an \nfa
with the same transition function as $\mathcal{A}$, 
and all its states being initial. 
For $\suff{\comp{L}}$, the corresponding \nfa will be the same, 
but flipping the finality to all the states.  
Because both operands share the same structure,
the \dfa corresponding to the intersection
 will be the \dfa resulting from the subset construction 
 considering a suitable set of final states
(they must contain at least one final state and a non-final one). 
As all states $S\subseteq 2^Q$ with $|S|=1$ are either final or non-final, 
they cannot be useful, therefore they can be ignored.
\end{proof}

Let ${\cal A}=(Q,\Sigma,\delta,i,F)$ be the minimal \dfa recognizing
$L$ with $|Q|=n$. 
Let $Q=\{0,\ldots,n-1\}$ and $R_i $, $0\leq i\leq n-1$, be the
left-quotients of $L$ (possibly including the empty set). 
From  Theorem~\ref{theo:dist-quot}, we have:
  \begin{equation}
\label{eq:atoms}
    \dis{L}=\bigcup_{i\in Q}R_i  \setminus
    \bigcap_{i\in Q}R_i =\left(\bigcup_{i\in
        Q}R_i \right)\cap\comp{\left (\bigcap_{i\in 
          Q}R_i \right )} = \comp{\left(\bigcap_{i\in
          Q}\comp{R_i }\right)\cup \left (\bigcap_{i\in 
          Q}R_i \right )}. 
  \end{equation}

In the following we identify the states of ${\cal A}$ with the
corresponding left-quotients. 
Instead of using traditional techniques to prove the correctness of 
tight upper bounds of operational state complexity, 
here we consider a method based on the \emph{atoms} 
of regular expressions. 
Using this approach, we  aim to provide yet another piece of evidence for their
 broad applicability.

Brzozowski and Tamm introduced
the notion of atoms of regular languages in~\cite{brzozowski11:_theor_of_atomat}
  and studied their state complexity in~\cite{brzozowski13:_compl_of_atoms_of_regul_languag}. 
An \emph{atom} of a regular language $L$ with $n$ quotients $R_0$, \ldots,
$R_{n-1}$ is a non-empty intersection $K_0\cap\cdots\cap K_{n-1}$,
where each $K_i$ is a quotient $R_i $, or its complement
$\comp{R_i }$. Atoms of $L$ are partition of  $\Sigma^\star$.
In particular, $A_Q=\bigcap_{i\in Q}R_i $ 
($A_\emptyset=\bigcap_{i\in Q}\comp{R_i }$) is an atom with zero complemented 
(uncomplemented) quotients.
In~\cite{brzozowski13:_compl_of_atoms_of_regul_languag} it
was proved that the state complexity of both those atoms is $2^n-1$. 
Using similar arguments, we prove the following theorem.

\begin{theorem}
  \label{theo:scdistub}
If a regular language $L$ has a  minimal \dfa with $n\geq 2$ states, then
$sc(\dis{L})\leq 2^n-n$. 
\end{theorem}
\begin{proof}
  Let ${\cal A}=(Q,\Sigma,\delta,i,F)$ be the minimal \dfa recognizing
  $L$ with $|Q|=n$. 
Then $Q=\{0,\ldots,n-1\}$, and let $R_i $, $0\leq i\leq n-1$ be the
  (left-)quotients of $L$. 
\ Using Equation~(\ref{eq:atoms}), every quotient $w^{-1}\dis{L}$
  of $\dis{L}$, for $w\in \kleene{\Sigma}$, is given by:
\begin{equation*}
  w^{-1}\dis{L} = \left(\bigcup_{i\in
      Q}w^{-1}R_i \right)\cap\comp{\left (\bigcap_{i\in 
        Q}w^{-1}R_i \right )},
\end{equation*}
\noindent where all $w^{-1}R_i$, $0\leq i\leq n-1$, are also
quotients of $L$, and they may not be distinct. 
Considering all non-empty subsets of quotients of $L$, 
there would be at most $2^n$ quotients of $\dis{L}$. 
However, all subsets, $R_j$, with exactly one element 
will lead to the empty quotient.  Thus, $sc(\dis{L})\leq 2^n-n$. 
\end{proof}

Brzozowski~\cite{brzozowski13:_in_searc_of_most_compl_regul_languag}
presented a family of languages $U_n$ which provides witnesses for the
state complexity of several individual and combined operations over
regular languages. Brzozowski and
Tamm~\cite{brzozowski13:_compl_of_atoms_of_regul_languag} proved that
$U_n$ was also a witness for the worst-case state complexity of
atoms. This family is defined as follows. For each $n\geq 2$, 
we construct the \dfas $D_n = (\Set{0,\ldots,n-1}, \Set{0,1,2},
\delta,0,\{n-1\})$, where  
$\delta(i,0)=i+1\mod n$, $\delta(0, 1) = 1$, $\delta(1, 1) = 0$,
$\delta(i, 1) = i$ for  
$i > 1$, $\delta(i, 2) = i$ for $0 \leq i\leq n-2$, and $\delta(n-1,
2) = 0$. 
We denote by $U_n$
the language accepted by $D_n$, i.e.,
\begin{equation}
 \label{edef:universal}
U_n=\lang{D_n}.
\end{equation}
We show that $U_n$ is also a witness for the lower-bound of the state
complexity of $\dis{L}$. 
First, observe that automata $D_n$, $n\geq 2$ are minimal. 
In Figure~\ref{fig:univ}, we present $D_4$.

\begin{figure}[h!]
\begin{center}
\begin{tikzpicture}[>=stealth, shorten >=1pt, auto, node distance=1.9cm,initial text={}]

\node[state, initial, inner sep=1pt, minimum size=7pt] (S0) {};
\node[state, inner sep=1pt, minimum size=7pt][right of=S0] (S1) {};
\node[state, inner sep=1pt, minimum size=7pt][right of=S1] (S2) {};
\node[state, accepting, inner sep=1pt, minimum size=7pt][right of=S2] (S3) {};
\path[->] (S0) edge [loop above] node {$c$} ()
               edge [bend left] node {$a,b$} (S1)
		  (S1) edge node {$a$} (S2)
		       edge node [pos=.2]{$b$} (S0)
			   edge [loop above] node {$c$} ()
		  (S2) edge node {$a$} (S3)
		       edge [loop above] node {$b,c$} ()
		  (S3) edge [loop above] node {$b$} ()
		       edge [in=340, out=200] node [pos=.2]{$a,c$} (S0);
\end{tikzpicture}
\end{center}
\caption{Universal witness $D_4$.} 
\label{fig:univ} 
\end{figure}
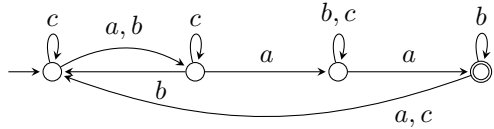
Next, we give the lower bound for the number of states of a 
\dfa accepting $\dis{U_n}$.
\begin{theorem}
    \label{theo:scdisttight}
For $n\geq 2$, the minimal \dfa accepting $\dis{U_n}$ has $2^n-n$ states.
\end{theorem}
\begin{proof}
  Let $A_n=(R_0\cap \ldots \cap R_{n-1})$, and 
$A_\emptyset=(\comp{R_0}\cap \ldots \cap \comp{R_{n-1}})$, 
be the two atoms of $U_n$ as above, where $R_i$ are its quotients $0\leq i\leq n-1$. 
Then $\dis{U_n}=\comp{A_n\cup A_\emptyset}$.  
Brzozowski and Tamm proved that
  $sc(A_n)=sc(A_\emptyset)=2^n-1$. 
Applying the construction given
  in Theorem~\ref{theo:algodis} to $\dis{U_n}$, and noting that a regular language
  and its complement have the same state complexity, we obtain the upper
  bound.
\end{proof}

If $sc(L)=1$, by Example~\ref{ex:singleton}, we have that $sc(\dis{L})=1$. If $L$ has $\emptyset$ as a quotient, by Lemma~\ref{lem:qempty},  the upper bound for $sc(\dis{L})$ coincides with the one for $\suff{L}$, i.e. it is $2^{n-1}$ if $sc(L)=n$,~\cite{brzozowski14:_quotien_compl_of_closed_languag}. This upper bound is achieved by the family of languages represented in Figure~\ref{fig:upsuffe}.

\begin{figure}[h!]
\begin{center}
\begin{tikzpicture}[>=stealth, shorten >=1pt, auto, node distance=1.9cm,initial text={}]

\node[state, accepting, initial, inner sep=1pt, minimum size=5pt] (S0) {$s_0$};
\node[state, inner sep=1pt, minimum size=5pt][right of=S0] (S1) {$s_1$};
\node[state, inner sep=1pt, minimum size=5pt][right of=S1] (S2) {$s_2$};
\node[right of=S2] (SS) {$\cdots$};
\node[state, inner sep=1pt, minimum size=5pt][right of=SS] (SN2) {$s_{n-2}$};
\node[state, inner sep=1pt, minimum size=5pt][above of=S0] (SN1) {$s_{n-1}$};

\path[->] (S0) edge node {$a$} (S1)
               edge node {$b$} (SN1)
		  (S1) edge node {$a$} (S2)
		       edge [in=345,out=195] node [pos=.2] {$b$} (S0)
		  (S2) edge node {$a$} (SS)
		       edge [loop above] node {$b$} ()
		  (SS) edge node {$a$} (SN2)
		  (SN2) edge [in=335, out = 205] node [pos=.15] {$a$} (S0) 
		       edge [loop above] node {$b$} ()
		  (SN1) edge [loop right] node {$a,b$} ();
\end{tikzpicture}
\end{center}
\caption{Witness family for  $sc(\dis{L})$ when $L$ has $\emptyset$ as a quotient.} 
\label{fig:upsuffe} 
\end{figure}
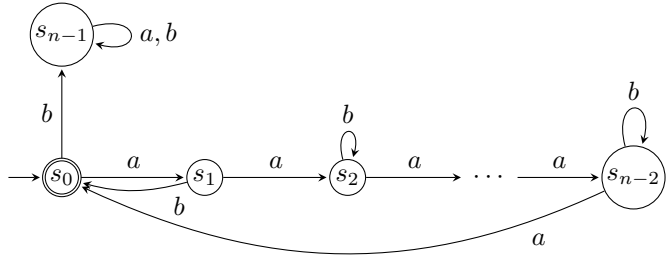
Having considered some  properties of the distinguishability language,
we would like  to select only the set of minimal words that
distinguishes between distinct quotients. Obviously, 
this is a subset of $\dis{L}$, and in the following  section we study its properties.

\section{Minimal Distinguishable Words}
\label{sec:minDist}
An even more succinct language distinguishing all different quotients
of a regular language, in fact a finite one, can be obtained if we
consider only the shortest word that distinguishes each pair of
quotients.  

\begin{definition}
\label{def:distmin}  
    Let $L$ be a regular language, and assume we have an order over 
the alphabet $\Sigma$.
If $x,y\in \Sigma^\star$ and $x\not\mneq{L} y$, 
 we define 
$$
\distminw{L}{x,y} = 
    \min\Set{w\mid w\in \disw{L}{x,y}},
$$
where minimum is considered with respect to the quasi-lexicographical order.
In case   $x\mneq{L} y$, $\distminw{L}{x,y}$ is undefined.
We can observe that if $x\not\mneq{L} y$, 
$\distminw{L}{x,y}=\min (x^{-1}L\Delta y^{-1}L)$.

 The set of minimal words distinguishing quotients of a language $L$ 
is 
$$
\distmin{L} = \Set{\distminw{L}{x,y} \mid  
                x,y\in \Sigma^*, x\not\mneq{L} y}.
$$
\end{definition}
 
%
\begin{example}
\label{ex:dismt} We present a few simple cases. Similar to the $\diso$ operator, 
we have the equalities:
$\distmin{\kleene{\Sigma}}=\distmin{\emptyset}=
\emptyset$ and $\distmin{\Set{\varepsilon}}=\{\varepsilon\}$. 
In case $a\in \Sigma$, $\distmin{a}=\dis{\Set{a}}=\Set{a,\varepsilon}$, and 
 $\distmin{\{a^n\}}=\dis{\Set{a^n}}=\Set{ a^i\mid 0\leq i\leq n}$, for $n\ge 2$.
\end{example}

\begin{example}
  \label{ex:dmin}
  Consider the language $L$ of Example~\ref{ex:d3}. 
We have the following equalities
$\distmin{L}=\Set{\varepsilon,0,1,01,11}$, 
$\distmin{\dis{L}}=\Set{\varepsilon,1,01,11}$, and  
$\distmin{\disn{L}{2}}=\Set{\varepsilon,1,11}$.
\end{example}

The previous example suggests that $\distmin{L}$ is also suffix closed.
   
\begin{theorem}
\label{theo:mindist_is_suffix_closed}
If $L$ is a regular language, then $\distmin{L}$ is suffix closed.
\end{theorem}
\begin{proof}
  Let $w\in \distmin{L}$, and let $w=uv$, with $u,v\in\kleene{\Sigma}$. 
Because $w\in \distmin{L}$,
  we can find two other words, $x,y\in \Sigma^\star$, such that
$xw\in L$ and $yw\notin L$, i.e.,
$xuv\in L$ and $yuv\notin L$. 
It follows that $v\in \disw{L}{xu,yu}$.
Since $v\in \disw{L}{xu,yu}$,  there exists 
$v'=\distminw{L}{xu,yu}$ and $v'\preceq v$.
Hence, $uv'\preceq uv$ and $uv'\in \disw{L}{x,y}$, which implies that 
$w=uv\preceq uv'$. Then we must have $uv'=uv$, which implies that  
$v=v'= \distminw{L}{xu,yu}\in  \distmin{L}$. 
\end{proof}

The next result gives an upper-bound for the number of elements of $\distmin{L}$.

\begin{theorem}
\label{theo:mindist_upp_bound_set}
If $L$ is a regular language with state complexity $n\geq 2$, 
then $|\distmin{L}|\leq n-1$.
\end{theorem}
\begin{proof}
For any three sets $A, B$ and $C$ we have the equality $(A\Delta B)\Delta(B\Delta C) = A\Delta C$.
Therefore, we can distinguish any pair from $n$ distinct sets with at most $n-1$ elements. 
To prove the theorem it is enough to choose the minimal words satisfying the above conditions,
 since the $n$ quotients of $L$ are all distinct (their symmetric difference is non-empty). 
\end{proof}

Now, we prove that the upper-bound is reached.
\begin{theorem}
\label{theo:mindist_bound_tight}
    The bound $n-1$ for the size of $\distmin{L}$, for a
    regular language $L$ with state complexity $n\geq 2$, is tight. 
\end{theorem}
\begin{proof}
Consider again the family of languages~$U_n$, 
described by  Equation~(\ref{edef:universal}). 
For each state $0\leq i\leq n-1$ of $D_n$, let $R_i$ be 
the corresponding quotient. 
It is easy to see that the minimal words for each quotient $R_i$ are $0^{n-i-1}$, 
and we can disregard the largest one. 
 \end{proof}

We now consider the iteration of the $\distmino$ operator. 
Because $\distmin{L}\subseteq \dis{L}$, $\distmin{L}\subseteq \suff{L}$,
 and $\distmin{L}$ is suffix closed, it follows that
  $\distminp{L}{2}\subseteq \distmin{L}$, and,  in general, 
\begin{equation}
\label{eSubsetDismin}
\distminp{L}{n+1}\subseteq \distminp{L}{n},\mbox{ for  all }n\geq 1. 
\end{equation}

By the finiteness of $\distmin{L}$, it follows 
 that there exists $n\geq 0$ such that $\distminp{L}{n+1}= \distminp{L}{n}$.
  For instance, considering the family of languages 
$U_n$ defined by equation (\ref{edef:universal}), we have that  $\distminp{U_n}{2}=\distmin{U_n}$.

Contrary to the hierarchy for $\dis{L}$, where the fixed point is reached for $n=2$,
 in the case of $\distmin{L}$ we have that  for any $n\geq 0$, 
there is a language for which the fixed point is reached after $n$ iterations of $\distmino$.

\begin{theorem}
\label{theo:fpdismin}
Given a regular language $L$ with state complexity $n$, 
the fixed point of $\distminp{L}{i}$, is reached for $0\leq i\leq  n-2$. 
\end{theorem}

\begin{proof}
Because $\distmin{L}$ is suffix closed, 
$\varepsilon\in \distminp{L}{i}$ for all every $i\geq 1$, thus
any automaton recognizing $\distminp{L}{i}$ has at least 2 states.
By Theorem~\ref{theo:mindist_upp_bound_set}, 
$\card{\distmin{L}}\leq n-1$. 
Using Equation~(\ref{eSubsetDismin})
we either have the same set, or a smaller set, thus we may lose 
at least one element at each iteration. Hence, $i\leq n-2$.
\end{proof}

If in the previous theorem we have established an upper-bound for the number of iterations 
of the $\distmino$ operator necessary to reach a fixed point, in the next one we show that the 
upper-bound can be reached.
 
\begin{theorem}
\label{theo:fpdismint}
For all  $n \geq 3$, there exists a regular language $L_n$, 
with $sc(L_n)=n$, such that
\begin{enumerate}[i)]
 \item $\distminp{L_n}{m-1}\neq \distminp{L_n}{m}$, for all $m<n-2$, and 
 \item $\distminp{L_n}{n-2}=\distminp{L_n}{n-1}$.
\end{enumerate}
\end{theorem}

\begin{proof}
Consider the family of languages 
$W_m=\suff{0^m1}=\Set{0^i1\mid 0\leq i\leq m}\cup \Set{\varepsilon}$, $m\geq0$. 
Then $sc(W_m)=m+3$ and
$\distmin{W_m}=\Set{0^i1\mid 0\leq i\leq m-1}\cup\Set{\varepsilon}=W_{m-1}$.
Because $W_0=\Set{1,\varepsilon}$ is a fixed point for $\distmino$, it follows that
\begin{enumerate}
 \item $\distminp{L_{m-3}}{m}\neq \distminp{L_{m-3}}{m-1}$, for all $m<n-2$, and 
  \item $\distminp{L_{m-3}}{n-2}=\distminp{L_{m-3}}{n-1}$.
\end{enumerate}
Hence, we can just take $L_n=W_{n-3}$.
\end{proof}

In the next section we use $\distmin{L}$ to recover $L$ as an $l$-cover language for 
$L\cap\Sigma^{\leq l}$.

\section{Using Minimal Distinguishability Words  to Recover the Original \mbox{Language}}
\label{sec:recoverL}

In Section~\ref{sec:intro}  we claim that for any regular language $L$ there exist a constant $l$, such that 
having the distinguishability language $\distmin{L}$,
we can recover the original language $L$, if we know all the words in $L$ of length less than or equal to $l$, i.e., 
the set $L\cap \Sigma ^{\leq l}$.

Thus, a positive learning procedure can be designed to recover
 the original language, $L$,
from every pair $(\distmin{L},l)$, such that $l$ is large enough.
It is obvious that if $L$ is a regular language, and ${\cal A}=(Q,\Sigma,\delta,q_0,F)$ 
a finite automaton recognizing $L$, i.e., $L=\mathcal{L}({\cal A})$, then $L$ is always 
an $l$-cover language for $L\cap \Sigma^{\leq l}$.
Thus, the goal is to determine the language $L$ as an unique $l$-cover language
for $L\cap \Sigma^{\leq l}$.

Let ${\cal A}=(Q,\Sigma,\delta,q_0,F)$ and $L$ be a regular language such that $L=\mathcal{L}({\cal A})$.
The automaton ${\cal A}$ is minimal if for any two states $p,q\in Q$, $p\neq q$,
we can find a word to distinguish between them. 
Hence,  ${\cal A}$ is minimal if we can find a word $w\in \distmin{L}$ such that it distinguishes between $p$ and $q$, 
thus the words $x_{\cal A}(p)$ and $x_{\cal A}(q)$ are distinguishable by some word in $w\in \distmin{L}$.

Let us consider the Myhill-Nerode equivalence induced by $L$, $\equiv_L$.
Two words $x_1$ and $x_2$ are equivalent, with respect to $\equiv_L$, if and only if there exists
$w\in \distmin{L}$ such that $x_1w\in L$ iff $x_2w\in L$.
Thus, by generating all words in the language of length
$\max(|x_1|,|x_2|)+\max\{|w| \mid w\in \distmin{L}\}$, we can decide 
after a finite number of steps if $x_1\equiv_L x_2$.

Therefore, the following algorithm can select all words $x_{\cal A}(p)$ for a minimal DFA
recognizing $L$:

\begin{algorithmic}[1]
\State $n \gets 1, x \gets \varepsilon, x_{\cal A}(n)\gets x, Q\gets \{n\}, l\gets 0$
\State $y\gets\Succ{x}\text{, where $\Succ{x}$ is the next word for the quasi-lexicographical order}$
\While {$|y|\leq l+1$} 
    \If {$y\not\equiv x_{\cal A}(q),\forall q\in Q$} 
        \State $n\gets n+1, x_{\cal A}(n)\gets y, Q\gets Q\cup\{n\}, l\gets|y|$
    \Else 
        \State {$\delta(p,a)\gets q\text{, where } y=za\equiv_L x_{\cal A}(q), z=x_{\cal A}(p), a\in\Sigma$}
    \EndIf
    \State $x\gets y$, $y\gets\Succ{x}$
    \While {$y$ is unreachable and $|y|\leq l+1$}\label{checkReachable} \Comment {i.e., $y=x_{\cal A}(q)wa$, $x_{\cal A}(q)w\equiv_Lx_{\cal A}(p)$, for some $p,q\in Q$,  $w\in \Sigma^+$, $a\in \Sigma$}
        \State $x\gets y$, $y\gets\Succ{x}$
    \EndWhile
\EndWhile
\State Set as final all $q\in Q$ such that $x_{\cal A}(q)\in L$
\end{algorithmic}


We observe that if all the words of length $|x|=l+1$ are equivalent 
with some shorter words,
then all the words of length  greater than $l+1$ will be declared unreachable by the previous algorithm, 
and no new transition can be added.
On the other hand, $\delta(q,a)$ is always defined, as $x_{\cal A}(q)a=x_{\cal A}(p)$, for some $p\not\equiv q$, or
$x_{\cal A}(q)a\equiv x_{\cal A}(p)$, and $\delta(q,a)=p$.
All the states in the above construction correspond to distinguishable words $x_{\cal A}(q)$,
 thus the DFA is minimal.

Note that when testing in quasi-lexicographical order if a new 
word in $\Sigma^{\star}$ is equivalent with previously distinct ones, 
will prune the branches corresponding to  equivalent words,
 and testing done in line~\ref{checkReachable} needs only to test if $y$ 
is on a cut-out branch.
Because all words of length greater than  $n=sc(L)$ must be equivalent
with some word of length less than $n$, it is enough to test only words
of length at most $sc(L)+1$. In order to conduct the equivalence 
test, it is enough to generate  all words of length $n+d+1$, where $d=\max\{|w| \mid w\in \distmin{L}\}$. 
Hence, all steps in the algorithm are well defined and they are only executed 
a finite number of times, thus it will produce a minimal DFA after a finite number of steps.


Therefore, we just proved the following theorem:
\begin{theorem}
 Let $L$ be a regular language and $L_D=\distmin{L}$. 
If we have an algorithm to generate all words in the language $L$ up to a given length $m$, 
then we have an algorithm to compute
\begin{enumerate}[i)]
  \item A number $l$, such that $L$ is an $l$-cover language for $L\cap\Sigma^{\leq l}$;
  \item A minimal DFA $A$ for $L$, which is at the same time, an $l$-DFCA for $L\cap\Sigma^{\leq l}$.
\end{enumerate}
\end{theorem}

A crucial role for producing the algorithm is the fact that testing the equivalence of two words
can be done using a finite number of steps, and this is possible if we know $\distmin{L}$.
However, it is not known if an equivalent procedure can be obtained if we know only $\dis{L}$, 
as it is possible to have two languages $L_1$,$L_2$
such that they share the same $\diso$ languages, but different  $\distmino$ languages.

For example, we can take $L_1=\{w\mid w=aaxba^n, x\in \{a,b\}^*, n\geq 0\}$
and $L_2=\{w\mid w=bbxab^n, x\in \{a,b\}^*, n\geq 0\}$.
We have that 
$\distmin{L_1} =\{ b, ab, e,aab\}$, $\distmin{L_2} = \{a,e,ba, bba\}$, 
and $\dis{L_1}=\dis{L_2}=\Sigma^\star$.
This example suggests why knowing $\dis{L}$ may not be enough.



In the next section we check under what conditions the results obtained 
so far can be generalized.

\section{Boolean Operations and Closure}
\label{sec:sufclosuregeneral}

In Section~\ref{sgenres} we used Boolean operations and the suffix
operation to compute the distinguishability language. The suffix
operation has the following properties:
\begin{enumerate}[i)]
\item $\suff{\emptyset}=\emptyset$;
\item $L\subseteq \suff{L}$;
\item $\suff{\suff{L}}= \suff{L}$;
\item $\suff{L_1\cup L_2}=\suff{L_1}\cup \suff{L_2}$, 
\label{P4}
\end{enumerate}
thus, it is a closure operator.
If we consider distinguishability operation as a unary operation 
on regular languages, we can see that it is obtained by applying 
finitely many times a closure operator and Boolean operations.
 The Closure-Complement Kuratowski Theorem \cite{Kuratowski1922} 
 says that using one set, one can obtain at most 14
 distinct sets using finitely many times one closure operator 
 and the complement operation. 
For the case of regular languages,  
Brzozowski et al. \cite{brzozowski11:_closur_in_formal_languag_and_kurat_theor}
determine the number of languages that can be obtained 
by applying finitely many times the  Kleene closure and complement.
However, the corresponding property \ref{P4} is not satisfied by the Kleene closure.

In this section we analyze the case of closure operators 
and Boolean operations, and ask if we apply them finitely 
many times we can still obtain only finitely many sets, or 
what is a necessary condition to obtain only finitely many sets.

In order to do this, we need to prove some technical lemmata. In the following $M$ denotes a nonempty set.

\begin{lemma}
\label{lfiniteuint}
 Let $N=\{\mathcal{A}_1,\ldots,A_n\}$ where $\mathcal{A}_i\in 2^M$, $1\leq i\leq n$.
 Then the free algebra $(N,\cup,\cap,\comp{\cdot},\emptyset,N)$ 
has a finite number of elements.
 \end{lemma}
\begin{proof}
 All expressions can be reduced to the disjunctive normal form, 
 and we only have finitely many such formulae.
\end{proof}

\begin{lemma}
\label{lcomute}
Let  $c_1,c_2:2^M\longrightarrow 2^M$ be
two closure operators such that $c_1$ and $c_2$ commute, i.e.,
$c_1\circ c_2=c_2 \circ c_1$.
Then the composition $c=c_2\circ c_1$ is also a closure operator.  
\end{lemma}
\begin{proof}
Let us verify the properties of a closure operator, thus if 
$L,L_1,L_2\in 2^M$, we have:
\begin{enumerate}[i)]
\item $c(\emptyset)=(c_2\circ c_1)(\emptyset)=c_2(c_1(\emptyset))=\emptyset$;
\item $L\subseteq c_2(L)\subseteq c_1(c_2(L))=(c_2\circ c_1)(L)=c(L)$;
\item $c(c(L))=(c_2\circ c_1)((c_2\circ c_1)(L))=
       (c_2\circ c_1)((c_1\circ c_2)(L))=
       (c_2\circ c_1\circ c_1\circ c_2)(L)=
       (c_2\circ c_1\circ c_2)(L)=
       (c_2\circ c_2\circ c_1)(L)=
       (c_2\circ c_1)(L)=c(L)$;
\item $c(L_1\cup L_2)=(c_2\circ c_1)(L_1\cup L_2)=
c_2(c_1(L_1\cup L_2))=
c_2(c_1(L_1)\cup c_1(L_2))=
c_2(c_1(L_1))\cup c_2(c_1(L_2))=
c(L_1)\cup c(L_2)$. 
\end{enumerate}
\end{proof}

In general, not all closure operators commute, for example, 
$N_0,N_1:2^\N\longrightarrow2^\N$ defined by 
$N_0(A)=A\cup \{x\in \N\mid  \exists k>0, x=2k\mbox{ if }2k-1\in A\}$,
$N_1(A)=A\cup \{x\in \N\mid  \exists k>=0, x=2k+1\mbox{ if }2k\in A\}$,
do not commute one with each other, as $N_0$ adds all the even numbers that
are successors of elements in the set $A$, and $N_1$  adds all the odd numbers
that are successors of elements in the set $A$. 
Applying the closure operators alternatively to a finite set $A$, 
we always obtain a new set. 

The next result is well known for the behaviour of closure operators when applied 
to a intersection of two other sets.

\begin{lemma}  
\label{lkcap}
Let $c$ be a closure operator on $2^M$.
 If $L_1,L_2\in 2^M$ are closed subsets,  then
$c(L_1)\cap c(L_2)=c(L_1\cap L_2)$.
\end{lemma}
Assume we have a finite number of sets $L_1, \ldots, L_m$.
Using closure and complement for each set $L_i$, 
$1\leq i\leq m$, we can obtain a 
finite number of sets \cite{Kuratowski1922}, say
  $M_1, \ldots,M_l$.
  Now consider a Boolean expression using $M_1, \ldots,M_l$.
Because we can transform all these Boolean expressions in
disjunctive normal form, the number of Boolean expressions
over $M_1, \ldots,M_l$ is finite.
Applying the closure operator to such an expression 
will commute with union, and the other sets 
are in the form 
$c(M_{i_1}\cap\cdots \cap M_{i_k})$.
If all $M_{i_j}$ $1\leq j\leq k$ are closed sets,
then   
$c(M_{i_1}\cap\cdots \cap M_{i_k})$ is a conjunction of 
some other sets $M_{j_1}, \ldots,M_{j_k}$, $1\leq j_i\leq l$.
Otherwise, if a set $M_{i_j}$ $1\leq j\leq k$ is not closed,
we may obtain new sets, as we can see from the following example:
if 
 $L_1=\{aaaa,abaabbaa,b\}$, $L_2=\{bbb,baaab,aa\}$, where $c=\suffo$,
then
 $c(L_1\cap c(L_2))\cap L_2\neq c(L_1)\cap L_2$.

This suggests that if a unary operation that combines Boolean 
operations and closure operators
is repeatedly applied to a set and we first apply the closure 
operator to the set and its complement, 
then we use other Boolean operations or the closure operator
finitely many times, we will always obtain  finitely many sets.
\ It follows that we have just proved the following lemma:
\begin{lemma}
\label{lem:fixgen}
Let  $c$ be a closure operator on $2^M$. If  $\bcopo:2^M\longrightarrow 2^M$ is defined as
a  reunion and intersections over $c(L)$ and $c(\comp{L})$, for $L\in 2^M$,
then 
\begin{enumerate}
\item for every set $A$, $\bcop{A}=c(B)$, for some $B\in 2^M$;
\item any iteration of $\bcopo$ will produce a finite number of sets;
\item if $\bcop{L} \subset L$, for all $L$, then $\bcopo$ has a fixed point. 
\end{enumerate}
\end{lemma}
 
Of course, if we have more than one closure operator, and we want to obtain 
finitely many sets,
we must first apply one closure operator to the collection of sets and their
 complements,
 then all the other Boolean operators and closure operator again.
In this way, we have guaranteed that we can only obtain finitely many sets.
In case the operation defined this way is monotone and bounded, it will have 
a fixed point.
In particular, we have a generalization of Theorem~\ref{theo:fixpoint}.
\begin{corollary}
\label{cor:ot}
Let $L\in 2^M$ and $c$ be a closure operator on $2^M$. 
If $\bcop{L}=c(L)\cap c(\comp{L})$ then $\bcopn{L}{3}=\bcopn{L}{2}$.
\end{corollary}


\section{More Distinguishability Operations}
\label{sec:otherdis}
A natural extension of $\diso$, as defined in Theorem~\ref{theo:dss}, is to consider prefix operator and infix operator, thus,
$\dpre{L}=\preff{L}\cap\preff{\comp{L}}$, or
$\dinf{L}=\inff{L}\cap\inff{\comp{L}}$, where $\preff{L}$ denotes the language of all prefixes of $L$ and $\inff{L}$ the language of all factors of $L$.
Because $\preffo$ and $\inffo$ are closure operators, $\dpreo$ and $\dinfo$  will share properties of  $\diso$. In particular $\dpre{L}$ is prefix-closed, $\dinf{L}$ is infix-closed, and both satisfy Corollary~\ref{cor:ot}. If $L$ is $\emptyset$ or $\Sigma^\star$, then $\dpre{L}=\dinf{L}=\emptyset$.


In the following subsections we briefly consider these operators.

\subsection{Right Distinguishability}
\label{sec:rightdis}
Given a regular language $L$, the (Myhill-Nerode) relation on $\Sigma^\star$, $x\mnleq{L} y$ if and only if $(\forall u)u\in \Sigma^\star, ux\in L \Leftrightarrow uy\in L$ is an equivalence relation with finite index and  left invariant. The \emph{right quotient} of $L$ by a word $u\in \Sigma^\star$ is the language $Lu^{-1}=\{x\in \Sigma^\star \mid xu \in L\}$ and corresponds to an equivalence class of $\mnleq{L}$,~\cite{champarnaud13:_two_sided_deriv_for_regul,sakarovitch09:_elemen_of_autom_theor}.
For $x,y\in \Sigma^\star$, we define $\dprew{L}{x,y} = L^{-1}x\Delta L^{-1}y$. Then, if we define
the \emph{right distinguishability language} of $L$ by 
\begin{equation}
\label{eq:disleftR}
\dpre{L}=\{w\mid \exists x,y\in \Sigma^\star \ (wx \in L\ \wedge \ wy
\notin L)\},
\end{equation}
\noindent it is immediate that 
\begin{equation*}
 \dpre{L} = \bigcup_{x\in\kleene{\Sigma}}Lx^{-1}\setminus
 \bigcap_{x\in\kleene{\Sigma}}Lx^{-1}, 
\end{equation*}
and
\begin{equation*}
 \dpre{L}=\preff{L}\cap\preff{\comp{L}}.
 \end{equation*}
For $u\in \Sigma^\star$, $(Lu^{-1})^R=(u^R)^{-1}L^R$, i.e., the right quotients of $L$ are exactly the reversals of the (left) quotients of $L^R$, which  correspond to  the atoms of $L$,~\cite{brzozowski11:_theor_of_atomat}. Thus, $\dpre{L}$ is the language of the words that distinguish between pairs of different atoms  of $L$.
We have that 
\begin{equation}
\label{eq:dpredis}
\dpre{L}=(\dis{L^R})^R,
\end{equation}
i.e., $\dis{L}^R=\dpre{L^R}$. 

\begin{lemma}
\label{lem:qepref}
Let $L$ be a regular language. If $L$ does not have $\emptyset$ as a quotient, then  $\dpre{L}=\preff{\comp{L}}$.
\end{lemma}
\begin{proof}
Because $\emptyset$ is not a quotient of $L$, we have
 $\preff{L}=\Sigma^\star$, therefore $\dpre{L}=\preff{\comp{L}}$.
\end{proof}

We have  $\dpre{L}=L$ if and only if $\dis{L^R}=L^R$. 
In particular, $L$ has $\emptyset$ as a right quotient if and only if  $L^R$ has $\emptyset$ as quotient.
The fact that  $L$ has an empty  right quotient does not imply that $L$ has an empty (left) quotient, as can be seen with $L=(a+b)^\star a$, where  $\dpre{L}=\Sigma^\star$.
 The results in~\ref{lem:prelemma} follow immediately from~\ref{lem:qempty}--\ref{cor:finite}. 

\begin{lemma}
\label{lem:prelemma}
Let $L$ be a regular language. Then the following statements hold true:
\begin{enumerate}[i)] 
\item If $L$ has $\emptyset$ as a right quotient, then  $\dpre{L}=\preff{L}$.
\item If $L$ is s prefix-closed and  $L$ has a $\emptyset$ as a right quotient, then $\dpre{L}=L$.
\item If  $\dpre{L}=L$, then $L$ has $\emptyset$ as a right quotient.
\end{enumerate}
\end{lemma}
\begin{example}
\label{ex:prefe2}
In Figure~\ref{fig:prefe2} one can see, from left to right, the minimal  \dfa accepting the language $L$, 
the language $\dpre{L}$, $\dpre{L}\not=\dpren{L}{2}$, and the language $\dpren{L}{2}=\dpren{L}{n}$, for $n\ge 3$.
\end{example}

\begin{figure}[h!]
\begin{center}
\begin{tikzpicture}[>=stealth, shorten >=1pt, auto, node distance=1.2cm,initial text={}]
\node[state, initial, accepting, inner sep=1pt, minimum size=7pt] (S0) {};
\node[state, inner sep=1pt, minimum size=7pt][above right of=S0] (S1) {};
\node[state, accepting, inner sep=1pt, minimum size=7pt][right of=S1] (S2) {};
\node[state, inner sep=1pt, minimum size=7pt][below right of=S0] (S3) {};
\node[state, inner sep=1pt, minimum size=7pt][right of=S3] (S4) {};
\path[->] (S0) edge node {$0$} (S1)
               edge [bend left] node {$1$} (S3)
          (S1) edge [loop above] node {$1$} ()
               edge [bend left] node {$0$} (S2)
          (S2) edge [bend left] node {$0$} (S1)
               edge [loop above] node {$1$} ()
          (S3) edge [bend left] node {$0$} (S0)
               edge node {$1$} (S4)
          (S4) edge [loop above] node {$0,1$} ();
\end{tikzpicture}\quad\begin{tikzpicture}[>=stealth, shorten >=1pt, auto, node distance=1cm,initial text={}]
\node[state, initial, accepting, inner sep=1pt, minimum size=7pt] (S0) {};
\node[state, accepting, inner sep=1pt, minimum size=7pt][above right of=S0] (S1) {};
\node[state, accepting, inner sep=1pt, minimum size=7pt][below right of=S0] (S3) {};
\node[state, inner sep=1pt, minimum size=7pt][right of=S3] (S4) {};
\path[->] (S0) edge node {$0$} (S1)
               edge [bend left]node {$1$} (S3)
		  (S1) edge [loop above] node {$0,1$} ()
		  (S3) edge [bend left] node {$0$} (S0)
		       edge node {$1$} (S4)
		  (S4) edge [loop above] node {$0,1$} ();
\end{tikzpicture}\quad\begin{tikzpicture}[>=stealth, shorten >=1pt, auto, node distance=1.2cm,initial text={}]
\node[state, initial, accepting, inner sep=1pt, minimum size=7pt] (S0) {};
\node[state, accepting, inner sep=1pt, minimum size=7pt][above right of=S0] (S1) {};
\node[state, inner sep=1pt, minimum size=7pt][below right of=S1] (S2) {};
\path[->] (S0) edge [bend left] node {$1$} (S1)
               edge [bend right] node [swap] {$0$} (S2)
		  (S1) edge node [pos=.2]{$0$} (S0)
		       edge node {$1$} (S2)
		  (S2) edge [loop right] node {$0,1$} ();
\end{tikzpicture}
\end{center}
\caption{Automata for the languages $L$, $\dpre{L}$, and $\dpren{L}{n}$, $n\ge 2$.}

\label{fig:prefe2} 
\end{figure}

\begin{corollary}
\label{cor:dprefinite}
  If $L$ is a finite language, then $\dpre{L}=\preff{L}$.
\end{corollary}

The state complexity of the $\dpreo$ operation is given by the following theorem.
\begin{theorem}
\label{theo:scdpre}
If $L$ is recognized by a minimal \dfa with $n\geq 2$ states, then
 $sc(\dpre{L})=n$.
\end{theorem}
\begin{proof}
 If both $L$ and $\comp{L}$ do not have $\emptyset$ as a quotient, then $\dpre{L}=\Sigma^\star$ 
and only one state is needed for a \dfa accepting $\dpre{L}$. 
Otherwise, let $\mathcal{A}=(Q,\Sigma,\delta,i,F)$ be the minimal \dfa recognizing $L$ with $|Q|=n$.
We have that at least one of $\mathcal{A}$ or $\comp{\mathcal{A}}$ has a dead state. 
To obtain a \dfa for $\preff{L}$ one needs only to  consider all states of $\mathcal{A}$ final, except the dead state, if it exists. 
To get a \dfa for $\dpre{L}$, we also need to exclude from the set of final states the possible dead state of the \dfa $\comp{\mathcal{A}}$, recognizing $\comp{L}$, which coincides with $\mathcal{A}$, except that the set of final states is $Q\setminus F$. Tightness is achieved  for the family of languages $L_n=\{a^i\mid i\leq n-2\}$, which are prefix closed,~\cite{brzozowski14:_quotien_compl_of_closed_languag}.\end{proof}

In Section~\ref{sec:minDist}, we considered the language of the shortest words that distinguish pairs of left quotients of $L$, $\distmin{L}$. 
In this case, we can define $\dpremin{L}=\Set{\dpreminw{L}{x,y} \mid  x\not\mnleq{L} y}$, where
$\dpreminw{L}{x,y} = \min\Set{w\mid w\in \dprew{L}{x,y}}
$ if $x\not\mnleq{L} y$, and minimum is considered with respect to the quasi-lexicographical order. 
Using Equation (\ref{eq:dpredis}), one can have a finite set of words that distinguish between right quotients, namely $(\distmin{L^R})^R$.
However, using the notion of atoms we can compute directly $\dpremino$.
As we seen before, $\dpre{L}$ distinguishes between pairs of atoms of $L$. To estimate the number of elements of $\dpremin{L}$, we recall the relation between atoms and right quotients.

 Let ${\cal A}=(Q=\{0,\ldots,n-1\},\Sigma,\delta,i,F)$ be the minimal \dfa recognizing $L$ and let $R_i$, $0\leq i\leq n-1$ be the left quotients of $L$. 
Each atom can be characterized by a set $S\subseteq Q$ such that 
$A_{S}=\bigcap_{i\in S}R_i\bigcap \bigcap_{i\notin S}\comp{R_i}$.
Every  $x\in \Sigma^\star$ belongs exactly to one atom $A_{S_x}$, and if $x\mnleq{L} y$, i.e, $Lx^{-1}=Ly^{-1}$, 
then $x$ and $y$ belong to the same atom. 
Thus, the minimal word that distinguishes two distinct right quotients with correspondent  sets $S$ and $S'$ is 
$$\min\{w \mid w\in L_{S} \notequiv w\in L_{S'}\},$$ 
where $L_T=\bigcup_{i\in T}L_i$ for $T \subseteq Q$ and $\min\{w\mid w \in w\in L_T\}=\min\{x_{\cal A}(i)\mid i\in T\}$. 
Therefore $\dpreminw{L}{x,y}=\min\{x_{\cal A}(i)\mid i \in S_x\Delta S_y\}$. 
Using Theorem~\ref{theo:mindist_upp_bound_set}, it follows that $|\dpremin{L}|\leq n -1$. 
We also have that $\dpremino$ is prefix closed, $\dpreminp{L}{n+1}\subseteq \dpreminp{L}{n},\mbox{ for  all }n\geq 1$,  
and we can reach the fixed point in maximum $n-2$ iterations, using Theorem~\ref{theo:fpdismin} and 
Theorem~\ref{theo:fpdismint}, with  $W_n=\preff{10^n}$.

\subsection{Two-sided Distinguishability}
\label{sec:twosideddis}
Given a language $L\subseteq \Sigma^\star$, we can define the (Myhill-Nerode) equivalence relation on $\Sigma^\star\times \Sigma^\star$, 
$(x,y)\not\mntseq{L}(x',y')$ if and only if 
$(\forall u)u\in \Sigma^\star,xuy\in L \Leftrightarrow x'uy'\in L$. If $L$ is regular, $\mntseq{L}$ is of finite index and for $u,v\in \Sigma^\star$, the two-sided quotient $u^{-1}Lv^{-1}=\{x\in \Sigma^\star \mid uxv \in L\}$ corresponds to an equivalence class of $\mntseq{L}$.
We note that $u^{-1}Lv^{-1}=(u^{-1}L)v^{-1}=u^{-1}(Lv^{-1})$.
Two-sided quotients were recently used to define biautomata, which deterministic versions recognize exactly regular languages,~\cite{holzer13:_nondet_biaut_and_their_descr_compl,klima12:_biaut}, and \emph{couple} \nfas,
  which can recognize linear languages,~\cite{champarnaud13:_two_sided_deriv_for_regul}.

We define
the \emph{two-sided distinguishability language} of $L$ by 
\begin{equation}
\label{eq:disleftB}
\dinf{L}=\{w\mid \exists x,y,x',y'\in \Sigma^\star \ (xwy \in L\ \wedge \ x'wy'
\notin L)\}.
\end{equation}
\noindent It is immediate that 
\begin{equation*}
 \dinf{L} = \bigcup_{x,y\in\kleene{\Sigma}}x^{-1}Ly^{-1}\setminus
 \bigcap_{x,y\in\kleene{\Sigma}}x^{-1}Ly^{-1}, 
\end{equation*}
and
\begin{equation*}
 \dinf{L}=\inff{L}\cap\inff{\comp{L}}.
 \end{equation*}

Please note that  for all regular languages $L$, $\dis{L}\subseteq \dinf{L}$ and  $\dpre{L}\subseteq \dinf{L}$. 
If $L$ has $\emptyset$ as a (left) quotient, then $\emptyset$ is also a two-sided quotient. 
The following lemmata show that $\dinfo$ is always an $\inffo$ operation.

\begin{lemma}
\label{lem:qinf}
Let $L$ be a regular language. If $L$ does not have $\emptyset$ as a quotient, $\dinf{L}=\inff{\comp{L}}$.
\end{lemma}
\begin{proof} 
Since $\emptyset$ is not a quotient of $L$, it follows that  $\inff{L}=\Sigma^\star$, therefore $\dinf{L}=\inff{\comp{L}}$.
\end{proof}

\begin{lemma}
\label{lem:qeinf}
Let $L$ be a regular language. If $L$ has $\emptyset$ as a  quotient, then  $\dinf{L}=\inff{L}$.
\end{lemma}
\begin{proof}
We know that $\suff{\comp{L}}=\Sigma^\star$, thus $\inff{\comp{L}}=\Sigma^\star$.
\end{proof}

If $L$ is infix closed, then $L$ is also suffix and prefix closed. 
Excluding $\Sigma^\star$, the fixed points of $\dinfo$ are exactly the infix-closed languages. 
To see that, by the previous lemma we have:

\begin{lemma}
\label{lem:inffixpoint}
If $L$ is a infix-closed regular language and  $L$ has a $\emptyset$ as a quotient, then $\dinf{L}=L$.
\end{lemma}

\begin{lemma}
\label{lem:dinfpe}
Let $L$ be a regular language. If  $\dinf{L}=L$, then 
$L$ has $\emptyset$ as a quotient.
 \end{lemma}
\begin{proof}
Assume $L$ does not have $\emptyset$ as a quotient.
By Lemma~\ref{lem:qinf}, it follows  $\dinf{L}=\inff{\comp{L}}$, thus $L$ would not be a fixed point of $\dinfo$.
\end{proof}

From these two lemmata, one has
\begin{theorem}
\label{theo:dinffixpoint}
If $L$ is a regular language different from $\Sigma^\star$, $\dinf{L}=L$ if and only if $L$ is infix closed.
\end{theorem}

\begin{corollary}
\label{cor:dinffinite}
  If $L$ is a finite language, then $\dinf{L}=\inff{L}$.
\end{corollary}

In case $L=\Sigma^\star$, $\dinf{L}=\Sigma^\star\cap \emptyset=\emptyset$. Because,
$\dinf{\emptyset}= \emptyset\cap\Sigma^\star=\emptyset$, we have $\dinfn{L}{2}=\dinf{L}$.
This result can be generalized for all regular languages $L$, such that $\dinf{L}\neq\Sigma^\star$.

\begin{corollary}
\label{cor:f2f}
Given  a regular language $L$, if $\dinf{L}\neq\Sigma^\star$, then $\dinfn{L}{2}=\dinf{L}$. 
\end{corollary}
\begin{proof}
If $\dinf{L}\neq\Sigma^\star$, then either $L$ or $\comp{L}$ has $\emptyset$ as a quotient.
Hence,   $\dinf{L}=\inff{L}$ and $\dinfn{L}{2}=\inff{\inff{L}}=\dinf{L}$ or
 $\dinf{L}=\inff{\comp{L}}$ and  $\dinfn{L}{2}=\inff{\dinf{L}}=\dinf{L}$.
\end{proof}

The state complexity of $\dinfo$ coincides with the state complexity of the $\inffo$ operation.

\begin{theorem}
\label{theo:scdinf}
If $L$ is recognized by a minimal \dfa with $n\geq 2$ states, then
$sc(\dinf{L})=2^{n-1}$.
\end{theorem}

\begin{proof}
 If both $L$ and $\comp{L}$ do not have $\emptyset$ as a quotient, then $\dinf{L}=\Sigma^\star$, 
therefore only one state is needed for a \dfa accepting $\dinf{L}$. 
Otherwise, let $\mathcal{A}=(Q,\Sigma,\delta,i,F)$ be the minimal \dfa recognizing $L$ with $|Q|=n$, 
hence at least one of $\mathcal{A}$ or $\comp{\mathcal{A}}$ has a dead state.
An \nfa recognizing $\inff{L}$ can be obtained by marking as initial and final all states of $Q$ and deleting the possible dead states.
The correspondent \dfa has at most $2^{n-1}$ states,~\cite{brzozowski14:_quotien_compl_of_closed_languag}. 
An analogous construction can be used for $\inff{\comp{L}}$. 
Considering Lemma~\ref{lem:qinf} and Lemma~\ref{lem:qeinf}, a \dfa for $\dinf{L}$ is one of the above. 
Tightness is achieved for the family of languages recognized by \dfas represented in Figure~\ref{fig:upsuffe}.
\end{proof}

In this case, we can also define
$\dinfmin{L}=\Set{\dinfminw{x,y}{L} \mid  x\not\mntseq{L} y}$,
where 
$\dinfminw{x,y}{L} = \min\Set{w\mid w\in \dinfw{L}{x,y}}$. Although it is easy to see that $\dinfmin{L}$ enjoy similar properties of $\distmin{L}$ and $\dpremin{L}$, we leave open how to compute this set.


\section{Conclusion}
\label{sconc}

In this paper we have introduced two new operations on regular languages
that help us distinguish non-equivalent words under Myhill-Nerode equivalence.
The first one $\diso$ finds all these words and the second one $\distmino$
produces only the minimal ones, where minimum is considered with respect 
to the quasi-lexicographical order.
Both have fixed points under iteration. 
The number of iterations  until a fixed point is reached is bounded by $2$ for
the case of $\diso$, and it is bounded by the state complexity 
of the starting language for $\distmino$.
A full characterization of the fixed points of $\diso$ is provided. 
Brzozowski's universal witness $U_n$ reaches the upper-bound of $2^n-n$,
for the state complexity of $\diso$.
In the case of $\distmino$ operation, the maximum number of words in 
the language is $n-1$, where $n$ is the state complexity of the original language.
We used $\distmino$ to recover the original language $L$ as an $l$-cover language of 
an initial segment of the language, where words have length at most $l$, by generating 
words in the language up to length $l+d$, where $d$ is the length of the 
longest word in $\distmin{L}$.
We have generalized some results for these type of operations with arbitrary closures and 
Boolean operations.
We have extended the study to infix and prefix operators to distinguish right quotients and 
atoms of a language.

As open problems and future work we can consider the state complexity of combined
operations, when one of them is in the set $\{\diso,\dpreo,\dinfo\}$.
It worth mentioning that recovering the whole language from a finite number of words in the language
 is very useful in learning algorithms, thus  it would be useful to study all 
conditions that can help us to reconstruct it if we know some of the distinguishability languages.
Finite languages have the particularity that distinguishability operation reduces to the suffix one.
What would be corresponding operation for cover automata for finite languages and dissimilarity operation? 

\newcommand{\etalchar}[1]{$^{#1}$}

\end{document}